\documentclass[conference,10pt]{IEEEtran}

\usepackage{amsmath,amssymb,amsfonts,mathrsfs,bm}
\usepackage{amstext}
\usepackage{upgreek}
\usepackage{multicol}
\usepackage{indentfirst}
\usepackage{graphicx}
\usepackage{paralist}
\usepackage{multirow}
\usepackage{tabularx}

\usepackage{cite}
\usepackage{citesort}
\usepackage{booktabs}
\usepackage{subfigure}
\usepackage{flushend} 
\newtheorem{proposition}{\bf Proposition}

\newtheorem{proof}{Proof}
\IEEEoverridecommandlockouts
\allowdisplaybreaks
\begin{document}
\title{Cache-enabled Heterogeneous Cellular Networks: Comparison and Tradeoffs}
\author{	
 	\IEEEauthorblockN{{\large Dong Liu and Chenyang Yang}}\\
	\IEEEauthorblockA{Beihang University, Beijing, China\\
		Email: \{dliu, cyyang\}@buaa.edu.cn}
\thanks{This work was supported in part by National Natural Science Foundation of China (NSFC) under Grant 61120106002 and National Basic Research Program of China (973 Program) under Grant 2012CB316003.}
} \maketitle

\begin{abstract}
Caching popular contents at base stations (BSs) is a promising
way to unleash the potential of cellular heterogeneous networks
(HetNets), where backhaul has become a bottleneck. In this paper, we
compare a cache-enabled HetNet
where a tier of multi-antenna macro BSs is overlaid by a tier of helper nodes having caches but no backhaul
with a conventional HetNet where the macro BSs tier is overlaid by
a tier
of pico BSs with limited-capacity backhaul. We resort stochastic  geometry  theory to derive the
area spectral efficiencies (ASEs) of these two kinds of HetNets and obtain the closed-form
expressions under a special case. We use numerical results to show
that the helper density is only 1/4
of the pico BS density to achieve the same target ASE, and the helper density can be further reduced 
 by increasing  cache capacity. With given total cache capacity within an area,
there exists an optimal helper node density that maximizes the ASE.
\end{abstract}

%

\section{Introduction}
To support the 1000-fold higher throughput in the fifth-generation (5G)
cellular systems, a promising way is to densify the network
by deploying more small base stations (BSs) in a macro cell \cite{densification}. Such
heterogeneous networks (HetNets)  can increase the area spectral efficiency (ASE)
\cite{Heterogeneous}, which largely relies on
high-speed backhaul links. Although optical
fiber can provide high capacity, bringing fiber-connection to every single
small BS is rather labor-intensive and expensive. Alternatively, digital
subscriber line (DSL) or microwave backhaul may easily become a bottleneck
and frustratingly impair the throughput gain brought by the network densification
\cite{Femtocell}.

Recently, it has been observed that a large portion of mobile data
traffic is generated by many duplicate downloads of a few popular contents
\cite{woo2013comparison}. On the other hand, the storage capacity of today's
memory devices grows rapidly at a relatively low cost. Motivated by these facts, the authors  in \cite{Andy2012}
 suggested to replace small BSs by the BSs  that have weak backhaul links (or even completely without backhaul)
but have high capacity caches, called \emph{helper nodes}. By optimizing the caching policies to serve
more users under the constraints of file downloading time, large throughput
gain was reported. Considering small cell networks (SCNs) with backhaul of
very limited capacity, the
authors in \cite{Procach14} observed that the backhaul traffic load can be
reduced by caching files at the small BSs  based on their popularity. These results indicate that by fetching
contents locally instead of fetching from core network via backhaul links redundantly,
equipping caches at BSs is a promising way to unleash the
potential of HetNet.

Nonetheless, the performance gain of a cache-enabled HetNet over a
conventional HetNet with limited-capacity backhaul is still unknown. In \cite{GlobalSIP,Dong}, both the
throughput and energy efficiency of homogeneous cached-enabled cellular networks with hexagonal cells were analyzed. For HetNets or SCNs, however, it is more appropriate to use Poisson Point Process (PPP) to model the BS location \cite{flexible}. Stochastic geometry method
\cite{tractable} was first applied in \cite{EURASIP} for a homogeneous
cache-enable SCN where average delivery rate was
derived by assuming that the delivery rate is a constant when channel capacity exceeds a threshold. In \cite{pushing}, the throughput of a cache-enabled network with content pushing to users, device-to-device communication and caching at relays was derived, where every node (including the macro BS (MBS) and relay) is with a single antenna and with high-capacity backhaul. In
\cite{tony}, the file transmission success probability of cooperative transmission among helper nodes was analyzed, where the helpers and BSs are operated in orthogonal bandwidth.

In this paper, we investigate the benefits of cache-enabled
HetNet with respect to conventional HetNet with limited-capacity backhaul, and reveal
the tradeoff in deploying cache-enabled HetNet. We consider two kinds of
HetNets, where a tier of multi-antenna MBSs is overlaid with either
a tier of pico BSs (PBSs) with limited-capacity backhaul or a tier of helper nodes
with caches but without backhaul connection, and the two tiers are full frequency reused. We  derive the
average ASEs of the two kinds of HetNets respectively as functions of
BS/helper node density, user density, storage size, file popularity and
backhaul capacity, and obtain closed-form expressions of ASEs under
a special case. We first use simulations to validate our analysis. Then, we use
numerical results to show the merits of the cache-enabled HetNet: (1) It can double the ASE over
conventional HetNet with the same PBS/helper density. (2) The helper density is only a quarter of the PBS density to achieve the same target ASE, which
can reduce the cost of deployment and operation remarkably. Moreover, we
find that the helper density can be traded off by the cache capacity to achieve a target ASE. Given the total cache capacity within an
area, there exists an optimal helper density that maximizes the ASE.

\section{System model}
We consider two kinds of HetNets, as shown in
Fig. \ref{fig:layout}.

\begin{enumerate}
	\item Conventional HetNet: A tier of MBSs is overlaid with a tier of denser PBSs. The PBSs are connected to the core network via limited-capacity backhaul links.
	\item Cache-enabled HetNet: A tier of MBSs is overlaid with a tier of denser helper nodes. The helpers are not connected to core network via backhaul but have caches. Each helper can  cache $N_c$ files, which have been placed at the helper during
	off-peak times by broadcasting.
\end{enumerate}

\begin{figure}[!htb]
	\centering
	\subfigure[Pico BSs with backhaul]{
		\label{fig:pico} 
	\hspace{0mm}	\includegraphics[height=0.17\textwidth]{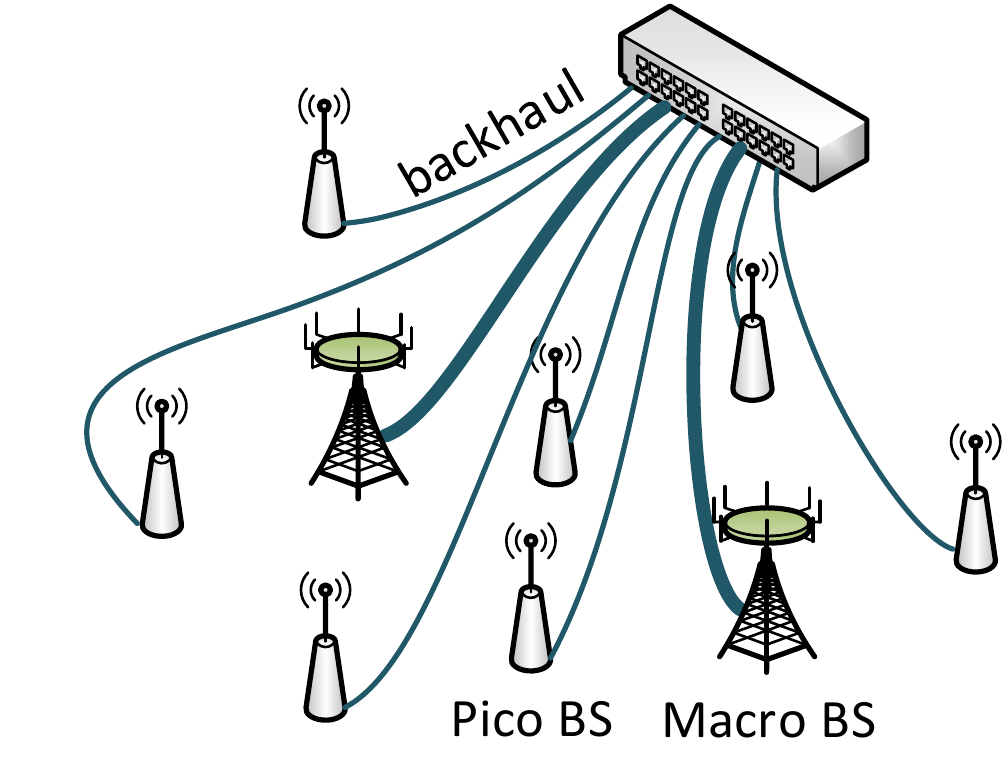} \hspace{2mm} }
	\subfigure[Helper nodes with caches]{
		\label{fig:helper} 
	\hspace{1mm}	\includegraphics[height=0.16\textwidth]{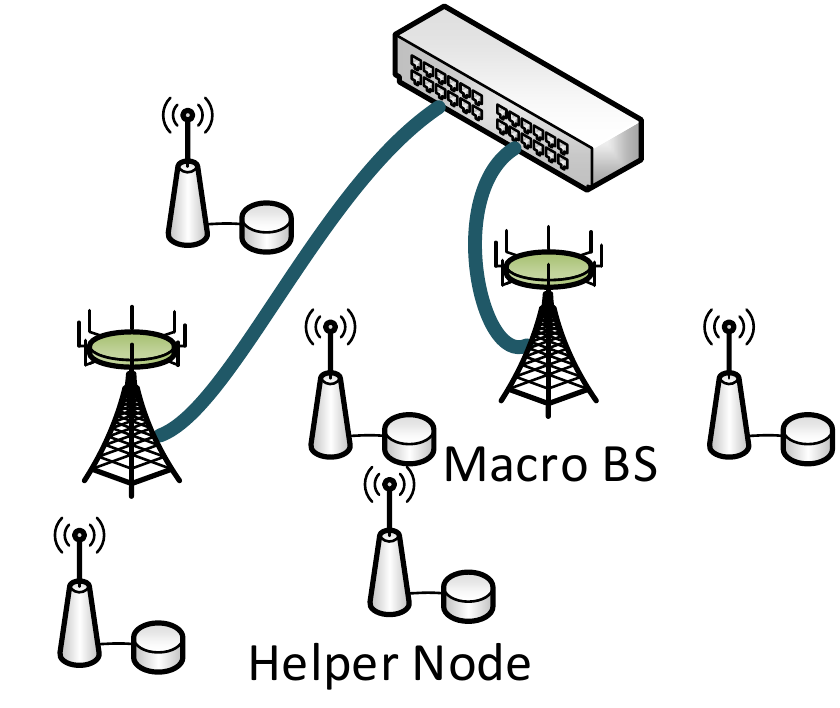}\hspace{2mm}}
	\caption{Layouts of the considered two kinds of HetNets.}
	\label{fig:layout} 
\end{figure}

The distribution of MBSs, PBSs/helpers and users are modeled
as three independent homogeneous PPPs, denoted as
$\Phi_1$, $\Phi_2$ and $\Phi_u$ with density $\lambda_1$, $\lambda_2$ and
$\lambda_u$, respectively. Each MBS is equipped with $M_1\geq 1$ antennas, and
each PBS or helper node is equipped with $M_2 = 1$ antenna. The transmit
power at each BS (or helper node) and pathloss exponent in the $k$th tier are $P_k$ and
$\alpha_k$, respectively. Since the MBSs are able to connect to the
core network via optical fiber with high capacity, while PBSs
usually employ cost-effective DSL or microwave backhaul with low capacity, the backhaul
capacity of each MBS is assumed as unlimited while the backhaul
capacity of each PBS is assumed as a finite value, $C_{\rm bh}$.\footnote{For example, the average download speed of current DSL in
USA is about $5$ Mbps \cite{golrezaei2013femtocaching}. For
notational simplicity, the backhaul capacity $C_{\rm bh}$ in our analysis is normalized
by the downlink transmission bandwidth $W$ in unit of nat/s/Hz, e.g., when $W = 20$ MHz with $10$ Mbps
backhaul, $C_{\rm bh} = 0.5$ bps/Hz $= 0.347$ nats/s/Hz.}

We assume that each user randomly requests a file from a static content
catalog that contains $N_f$ files. The files are indexed according to the
popularity, ranking from the most popular (the
1${st}$ file) to  the least popular (the $N_f$th file). The probability of requesting the $f$th file follows a Zipf-like
distribution \cite{breslau1999web}, i.e., $ p_f = {f^{-\delta}}\big/ {\sum_{n = 1}^{N_f}n^{-\delta}} $,
where the skew parameter $\delta$ determines the ``peakiness" of the
distribution, whose typical value is between 0.5 and 1.0. For mathematical simplicity, we assume that the
files are with unit size. Hence, the cache capacity of each helper node is
$N_c$.

We assume that $\lambda_u \gg \lambda_1$ such that each MBS has at least $M_1$ users to serve. Since the density of PBSs (or helper nodes) may become comparable with the density of users
along with the continuous network densification, some PBSs may have no users to serve. These inactive BSs will be
turned into idle mode to avoid interference. Each MBS
randomly selects $M_1$ users to serve at each time slot by zero-forcing
beamforming (ZFBF) with equal power allocation, while
each active PBS (or helper node) randomly selects one user to serve
at each time slot with full power.
These assumptions define a simple but typical scenario,
which however can capture the basic elements and reflect fundamental
tradeoffs.

Denote $k\in\{1,2\}$ as the index of the tier  which a randomly chosen user (called the
typical user) is
associated with, and denote $r_k$ as the distance between the user and its associated
BS $b_{k,0}$. The receive signal-to-interference-plus-noise ratio (SINR) of the typical user associated with BS $\!b_{k,0}$ is  
\begin{equation}
 \gamma_k \! = \! \frac{\frac{P_k}{M_k} h_{k,0} r_k^{-\alpha_k}}{\sum_{j=1}^{2} \sum_{i\in \tilde \Phi_j \backslash b_{k,0}}\!\! P_j h_{j,i} r_{j,i}^{-\alpha_j}  +  \sigma^2} \triangleq \frac{\frac{P_k}{M_k} h_{k,0} r_k^{-\alpha_k}}{I_k  +  \sigma^2} \hspace{-1mm}\label{eqn:gamma} 
\end{equation}
where $r_{j,i}$ is the distance between the typical user and the $i$th
active BS of the $j$th tier,  $\tilde \Phi_j$ represents the set of active
BSs in the $j$th tier, $h_{k,0}$ is the equivalent signal channel (including small-scale fading and precoding) from BS $b_{k,0}$ with unit mean, $h_{j,i}$ is the equivalent
interference channel from the $i$th BS in the $j$th tier, and $\sigma^2$ is the noise power. We consider Rayleigh
fading channels. Then, $h_{k,0}$ follows exponential distribution with unit mean
(i.e., $h_{k,0} \sim \exp(1)$), and  $h_{j,i}$ follows gamma distribution
with shape parameter $M_j$ and unit mean (i.e., $h_{j,i}\sim \mathbb{G}(M_j,
\frac{1}{M_j})$) \cite{adhoc}. Note that the distribution of $h_{j,i}$ is very different from that assumed in
\cite{pushing,tony} for cache-enabled networks, where
all the BSs have a single antenna. For notational simplicity, we define the interference power $I_k \triangleq \sum_{j = 1}^{2}I_{k,j} \triangleq \sum_{i\in \tilde \Phi_j \backslash b_{k,0}} P_j h_{j,i} r_{j,i}^{-\alpha_j}$.

\section{ASE of Conventional HetNet}
Consider that each user is associated with the BS with the strongest average receive power, say $P_{{\rm r}, j} = P_j {r}_j^{-\alpha_j}$ for the BS in the $j$th tier. For notational brevity, we define the normalized
transmit antenna number, transmit power and pathloss exponent
of the $j$th tier when the typical user is associated with
the $k$th tier as $ \widehat M_j = \frac{M_j}{M_k}$, $\widehat P_j
=\frac{P_j}{P_k}$, and $\widehat \alpha_j
=\frac{\alpha_j}{\alpha_k}$, respectively. Note that $\widehat M_k = \widehat
P_k =\widehat \alpha_k = 1$.

For notational simplicity, the data rate of the typical user is
expressed in unit of nats/s/Hz. The instantaneous
achievable rate of a typical user associated with the macro tier is 
\begin{equation}
R_1 = \ln(1+\gamma_1) \label{eqn:R1}
\end{equation}
Different from the macro tier, the achievable rate of a typical user
associated with the pico tier can not exceed the backhaul capacity, which is 
\begin{equation}
R_2 = \left\{\begin{array}{ll}
\ln(1+\gamma_2), & \text{if}~\ln(1+\gamma_2) \leq C_{\rm bh} \\
C_{\rm bh}, & \text{if}~ \ln(1+\gamma_2) >  C_{{\rm bh}}
\end{array}
\right. \label{eqn:R2} 
\end{equation}

In the sequel, we first
derive the average achievable rate of the typical user  associated with
the pico tier, which is 
\begin{equation}
 \bar R_2  = \mathbb{E}_{\gamma_2, r_2} [R_2] = \int_{0}^{\infty} \mathbb{E}_{\gamma_2}[R_2 | r ] f_{r_2}(r) dr \label{eqn:ERk} 
\end{equation}
where $\mathbb{E}_{\gamma_2}[R_2 | r ]$ is the conditional expectation of
$R_2$ conditioned on $r_2 = r$, $f_{r_k}(r) = \frac{2\pi\lambda_k}{\mathcal{P}_k}re^{-\pi\sum_{j=1}^{2}\lambda_j \widehat{P}_j {}^{{2}/{\alpha_j}} r^{{2}/{\widehat\alpha_j}} }$ is the probability density function (PDF) of the distance between the typical user and its serving BS, and $\mathcal{P}_k  = 2 \pi \lambda_k \int_{0}^{\infty} r e^{-\pi \sum_{j=1}^{2} \lambda_j \widehat{P}_j {}^{{2}/{\alpha_j}} r^{{2}/{\widehat\alpha_j}}}dr$ is the probability that the typical user is associated with the $k$th tier, which are given in Lemma 3 and Lemma 1 of \cite{flexible}, respectively.

Since $\mathbb{E}[X] = \int_{0}^{\infty}\mathbb{P}[X>x]dx$ for $X>0$ with $\mathbb{P}[\cdot]$ denoting probability, we
obtain 
\begin{align}
& \mathbb{E}_{\gamma_2} [R_2 | r ]\!  = \!\!\int_{0}^{\infty} \mathbb{P}[R_2 > x | r ] dx \overset{(a)}{=} \!\! \int_{0}^{C_{\rm bh}}  \mathbb{P}[\ln(1+\gamma_2) > x | r ] dx \nonumber\\
& \hspace*{-1mm} \overset{(b)}{=} \int_{0}^{C_{\rm bh}} \mathbb{E}_{I_2} [\mathbb{P}[h_{2,0} > M_2P_2^{-1}r^{\alpha_2} (I_2+\sigma^2) (e^x-1) | r, I_2 ]]dx \nonumber\\
& \hspace*{-1mm}\overset{(c)}{=} \int_{0}^{C_{\rm bh}} \mathbb{E}_{I_2} [e^{-M_2P_2^{-1}r^{\alpha_2}(I_2 + \sigma^2)(e^x-1)}] dx \nonumber\\
& \hspace*{-1mm}\overset{(d)}{=} \!\!\int_{0}^{C_{\rm bh}} \!  e^{-\frac{M_2}{P_2}r^{\alpha_2}  (e^x-1)\sigma^2} \! \prod_{j=1}^{2} \!\mathcal{L}_{I_{2,j}} \!\left(\tfrac{M_2 }{P_2}r^{\alpha_2} (e^x-1)\right) dx \label{eqn:ERkrk}
\end{align}
where step $(a)$ comes from the fact that $\mathbb{P}[R_2 > C_{\rm bh}] = 0$
due to \eqref{eqn:R2}, step $(b)$ is from \eqref{eqn:gamma} and using the law of total probability, step $(c)$ is
from $h_{k0} \sim \exp(1)$,  step $(d)$ follows because $\mathcal{L}_{\sum_{j}I_{k,j}}(s)
=\prod_{i} \mathcal{L}_{I_{k,j}}(s)$, and $\mathcal{L}(\cdot)$ denotes the Laplace
transform.

To derive the Laplace transform of
$I_{k,j}$, we model the distribution of active BS $\tilde{\Phi}_{k}$ as a homogeneous PPP
with density $p_{{\rm a},k}\lambda_k$ by thinning the BS distribution ${\Phi}_{k}$
as in  \cite{economy}, where $p_{{\rm a}, k}$ is the probability that a BS in the $k$th tier is active, which can be derived as 
\begin{equation}
 p_{{\rm a},k}  = 1 - \int_{0}^{\infty} e^{-\lambda_u x} f_{S_k}(x)dx \approx 1 - \left(1 + \tfrac{\mathcal{P}_k\lambda_u}{3.5\lambda_k}\right)^{-3.5} 
\end{equation}
where $f_{S_k}(x)$ is the PDF of the service area  $S_{k}$ of a BS in the $k$th tier, and the approximation comes from $	f_{S_k}(x) \approx \frac{3.5^{3.5}}{\Gamma(3.5)} \big(\tfrac{\lambda_k}{\mathcal{P}_k}\big){}^{3.5}x^{2.5}e^{-3.5{\lambda_kx}/{\mathcal{P}_k}}$ in \cite{offloading}, which is exact when the HetNet degenerates into a homogeneous
network. Note that for $k=1$ (i.e., the MBS tier), $p_{{\rm a},1}$ exactly equals to 1
since $\frac{\lambda_u}{\lambda_1} \to \infty$.
\begin{proposition}
The Laplace transform of the interference from the $j$th tier for a user
associated with the $k$th tier is 
\begin{multline}
\mathcal{L}_{I_{k,j}} (s) =
 \exp\Big(-\pi p_{{\rm a}, j}\lambda_j \widehat{P}_j{}^{\frac{2}{\alpha_j}}r^{\frac{2}{\widehat{\alpha}_j}} \times \\
 	 \big({}_2F_1 \big[-\tfrac{2}{\alpha_j}, M_j; 1-\tfrac{2}{\alpha_j}; -\tfrac{ sP_j }{ M_j\widehat{P}_j}r^{-\alpha_k}\big] - 1\big)\Big) \label{eqn:Laplace}
\end{multline}
where ${}_2F_1[\cdot]$ denotes the Gauss hypergeometric function.
\end{proposition}
\begin{proof}
	See Appendix A.
\end{proof}
With \eqref{eqn:Laplace} and \eqref{eqn:ERkrk}, the average rate in \eqref{eqn:ERk} becomes 
\begin{multline}
 \bar R_2 = \frac{2\pi\lambda_2}{\mathcal{P}_2} \int_{0}^{\infty}\!\! \int_{0}^{ C_{{\rm bh}}} \exp\bigg(- M_2 P_2^{-1} r^{\alpha_2} (e^x-1) \sigma^2\\
 - \pi \sum_{j=1}^{2} \lambda_j\widehat{P}_j{}^{\frac{2}{\alpha_j}}r^{\frac{2}{\widehat{\alpha}_j}}  \big( 1 + p_{a,j} \mathcal{Z}_j(x) \big) \bigg)  rdxdr  \label{eqn:ER2}
\end{multline}
where $\mathcal{Z}_j(x) \triangleq {}_2F_1 \big[-\tfrac{2}{\alpha_j}, M_j;
1-\tfrac{2}{\alpha_j}; \tfrac{(1 - e^x) }{\widehat M_j}\big] -
1$.

By letting $C_{\rm bh} \to \infty$, the average achievable rate
of a typical user associated with the macro tier can be similarly derived as
\begin{multline}
  \bar R_1 = \frac{2\pi\lambda_1}{\mathcal{P}_1} \int_{0}^{\infty}\!\! \int_{0}^{ \infty}\! \exp\bigg(-M_1 P_1^{-1}r^{\alpha_1}(e^x-1) \sigma^2 \\
  - \pi\sum_{j=1}^{2} \lambda_j \widehat{P}_j{}^{\frac{2}{\alpha_j}}r^{\frac{2}{\widehat{\alpha}_j}}  ( 1 + p_{a,j} \mathcal{Z}_j(x) ) \bigg)  rdxdr  \label{eqn:ER1}
\end{multline}

Since each active BS randomly selects $M_k$ users,
the average throughput of  a
randomly chosen active cell in the $k$th tier is $M_k\bar{R}_k$. The ASE, defined as the average throughput of a network per unit area
\cite{Quek}, can be obtained as 
\begin{equation}
  {\rm ASE} = \sum_{j=1}^{2}p_{{\rm a},j}\lambda_j M_j\bar{R}_j \label{eqn:ASE1}
\end{equation}
where $ p_{{\rm a}, j} \lambda_j $ is the density of active BSs in the $j$th
tier.

Although $\bar R_1$ and $\bar R_2$ can be numerically computed from \eqref{eqn:ER1} and
\eqref{eqn:ER2}, the computational complexity is very high. In the
following, we obtain closed-form expressions for approximated $\bar
R_{1}$ and $\bar R_{2}$ in a special case, which are accurate even for the general cases as illustrated by simulations later. 
\subsection{Special Case}
Since HetNets are usually interference-limited, it is reasonable to neglect the thermal noise \cite{flexible}, i.e., $\sigma^2 = 0$. Furthermore, we consider  equal path loss exponents of both tiers, $\alpha_j = \alpha$. Then, after a
change of variables $r^2 = v$, $\bar R_2$  in \eqref{eqn:ER2} can be further
derived as   
\begin{equation}
  \bar R_2 =  \frac{\lambda_2}{\mathcal{P}_2}\int_{0}^{ C_{{\rm bh}}} \bigg(\sum_{j=1}^{2} \lambda_j \widehat{P}_j{}^{\frac{2}{\alpha}} \big( 1  + p_{{\rm a},j} \mathcal{Z}_j(x) \big)\bigg)^{-1}dx  \label{eqn:sER2} 
\end{equation}
To derive
a closed-form expression, we first obtain the
approximation of $\mathcal{Z}_j(x)$ defined in  \eqref{eqn:ER2}. From the series-form expression
${}_2F_1 [a,b;c;z] = \sum_{n = 0}^{\infty} \frac{(a)_n
(b)_n}{(c)_n} z^n$, where $(x)_n \triangleq x(x+1) \cdots (x + n - 1)$
denotes the rising Pochhammer symbol,
we can approximate
$\mathcal{Z}_j(x)$ for small value of $x$ as 
\begin{align}
  \mathcal{Z}_j(x)&   = \frac{2}{2 - \alpha}\frac{M_j}{\widehat{M}_j}  (1 - e^x) + \mathcal{O}\left((e^x - 1)^2\right)  \nonumber\\
&   \overset{(a)}{=} \frac{2M_k }{\alpha - 2} x + \mathcal{O}(x^2) \approx \frac{2M_k}{\alpha - 2}  x \triangleq \mathcal{Z}_{j,low}(x) \label{eqn:2F1low} 
\end{align}
where step (a) is from $1-e^x = x + \mathcal{O}(x^2)$, and the approximation
is accurate when $1 - e^x \ll 1$, i.e., $x \ll \ln 2$. When the backhaul
capacity is very stringent, i.e., $C_{\rm bh} \to 0$, by substituting
\eqref{eqn:2F1low} into \eqref{eqn:sER2} and then using $\mathcal{P}_k = {\lambda_k}/{\sum_{j=1}^{2}\lambda_j
\widehat P_j{}^{\frac{2}{\alpha}}}$ derived from Lemma 1 in \cite{flexible}, $\bar R_2$ can be approximated as
\begin{align}
\bar R_2 \approx &   \frac{\lambda_2}{\mathcal{P}_2}\int_{0}^{ C_{{\rm bh}}} \bigg(\sum_{j=1}^{2} \lambda_j\widehat{P}_j{}^{\frac{2}{\alpha}}  \big( 1 + p_{a,j} \mathcal{Z}_{j,low}(x) \big) \bigg)^{-1} dx \nonumber \\
= &    \frac{\alpha-2}{2M_2}~ \mathcal{C}_1 \ln \left(1+ \frac{2M_2}{\alpha-2}\cdot\frac{C_{\rm bh}}{\mathcal{C}_1}\right)  \label{eqn:cER2} 
\end{align}
where $\mathcal{C}_1
\triangleq \big( {\sum_{j=1}^{2}\lambda_j \widehat P_j{}^{\frac{2}{\alpha}}}
\big) / \big( {\sum_{j=1}^{2}p_{{\rm a},j} \lambda_j
\widehat{P}_j{}^{\frac{2}{\alpha}}}\big)$.

Similar as deriving \eqref{eqn:sER2}, $\bar{R}_1$ in
\eqref{eqn:ER1} can be derived as 
\begin{equation}
  \bar R_1  = \frac{\lambda_1}{\mathcal{P}_1}\int_{0}^{\infty} \bigg(\sum_{j=1}^{2} \lambda_j \widehat{P}_j{}^{\frac{2}{\alpha}} \big( 1+ p_{{\rm a},j} \mathcal{Z}_j(x)\big)\bigg)^{-1}  dx  \label{eqn:sER1} 
\end{equation}
By using the following transformation \cite[eq. (9.132)]{jeffrey2007table},
\begin{align}
&\!\!\!\!\!{}_2F_1\left[a,b;c;z\right]  = \tfrac{\Gamma(c)\Gamma(b-a)}{\Gamma(b)\Gamma(c-a)} (-z)^{-a}{}_2F_1\left[a,a\!+\!1\!-\!c;a+\!1\!-\!b;\tfrac{1}{z}\right] \nonumber \\
&~~~~~~ +  \tfrac{\Gamma(c)\Gamma(a-b)}{\Gamma(a)\Gamma(c-b)} (-z)^{-b}{}_2F_1\left[b,b\!+\!1\!-c;b+\!1\!-\!a;\tfrac{1}{z}\right] \label{eqn:trans}
\end{align}
and  considering the
series-form expression of ${}_2F_1[\cdot] $, we can approximate
$\mathcal{Z}_j(x)$ for large value of $x$ as 
\begin{align}
 \mathcal{Z}_j(x)& = \tfrac{\Gamma\left(1-\frac{2}{\alpha}\right)\Gamma\left(M_j + \frac{2}{\alpha}\right)}{\Gamma(M_j)  } \big(\tfrac{e^x-1}{\widehat M_j}\big)^{\frac{2}{\alpha}} \!\!  - 1 + \mathcal{O}\left((e^x-1)^{-M_j}\right) \nonumber\\
&\approx \tfrac{\Gamma\left(1-\frac{2}{\alpha}\right)\Gamma\left(M_j + \frac{2}{\alpha}\right)}{\Gamma(M_j) \widehat M_j{}^{\frac{2}{\alpha}}} e^{\frac{2x}{\alpha}} - 1 \triangleq \mathcal{Z}_{j, high}(x) \label{eqn:2F1high} 
\end{align}
where $\Gamma (x)$ is the Gamma function, and the approximation is
accurate when $e^x - 1 \gg 1$, i.e., $x \gg \ln 2$.

By using the
approximation $\mathcal{Z}_j(x) \approx \mathcal{Z}_{j,low}(x)$ for $x\in
[0,\ln 2]$ and $\mathcal{Z}_j(x) \approx \mathcal{Z}_{j,high}(x)$ for $x\in
[\ln 2, \infty)$, we can approximate $\bar R_1$ as 
\begin{align}
 \bar R_1&   \approx   \frac{\lambda_1}{\mathcal{P}_1} \int_{0}^{\ln 2} \bigg(\sum_{j=1}^{2} \lambda_j \widehat{P}_j^{\frac{2}{\alpha}}  \left( 1 + p_{a,j}  \mathcal{Z}_{j,low}(x)\right) \bigg)^{-1} dx \nonumber \\
& ~~~   + \frac{\lambda_1}{\mathcal{P}_1} \int_{\ln 2}^{\infty} \bigg(\sum_{j=1}^{2} \lambda_j \widehat{P}_j^{\frac{2}{\alpha}} \left(1 + p_{{\rm a},j} \mathcal{Z}_{j,high}(x) \right) \bigg)^{-1}dx \nonumber \\
& =   \frac{\alpha-2}{2M_1}  \mathcal{C}_1 \ln \left(1+ \frac{2M_1}{\alpha-2}\!\cdot\! \frac{\ln 2}{\mathcal{C}_1}\right)  + \frac{\alpha}{2} \mathcal{C}_2  \ln \big(1 + \mathcal{C}_3 4^{-\frac{1}{\alpha}} \big) \label{eqn:cER1} 
\end{align}
where $\mathcal{C}_2 \triangleq \big( \sum_{j=1}^{2} \lambda_j
\widehat{P}_j{}^{\frac{2}{\alpha}} \big) / \big( \sum_{j=1}^{2} (1-p_{{\rm
a},j}) \lambda_j \widehat{P}_j{}^{\frac{2}{\alpha}} \big) $, $\mathcal{C}_3
\triangleq \big(\sum_{j=1}^{2} (1-p_{{\rm a},j}) \lambda_j
\widehat{P}_j{}^{\frac{2}{\alpha}} \big) / \big( \sum_{j=1}^{2} p_{{\rm
a},j} \lambda_j \widehat{P}_j{}^{\frac{2}{\alpha}} \mathcal{M}_j \big)$, and
$\mathcal{M}_j \triangleq {\Gamma(1-\frac{2}{\alpha})\Gamma(M_j +
\frac{2}{\alpha})}/\big({\Gamma(M_j)
\widehat{M}_j{}^{\frac{2}{\alpha}}}\big)$.

By substituting  \eqref{eqn:cER1} and \eqref{eqn:cER2} into \eqref{eqn:ASE1}, we can obtain the closed-form expression of an approximated ASE, which is accurate as shown by
simulation later.

\section{ASE of Cache-enabled HetNet}
We consider that each helper node caches the $N_c$
most popular files, which is the optimal caching policy in terms of cache hit-ratio when each user can only associate with one
node \cite{Andy2012}.

We call a user whose requested file is cached at the helper node as a
\textit{cache-hit} user and the others as \textit{cache-miss} users. The
probability that the typical user is a cache-hit user is 
\begin{equation}
 	p_{\rm h} = \sum_{f=1}^{N_c} p_f = \frac{\sum_{f=1}^{N_c} f^{-\delta}}{\sum_{n=1}^{N_f} n^{-\delta}} \label{eqn:ph} 
\end{equation}
Since the users are distributed as PPP and the file requests of the users
are independent and identical, the distribution of cache-hit users and
cache-miss users also follow PPPs,  respectively with density $p_{\rm h} \lambda_u$ and
$(1-p_{\rm h})\lambda_u$.

Considering that helpers are not connected with backhaul,
cache-miss users can only associate with the macro tier, while cache-hit
users can associate  either with macro or helper tier.

For the cache-hit users, the cell association is based on the maximal average receive power.
Then, the probability that a cache-hit user is associated with the
$k$th tier can be obtained from Lemma 1 in \cite{flexible} as 
\begin{equation}
  \mathcal{P}_{{\rm h},k} = 2 \pi \lambda_k \int_{0}^{\infty} r e^{-\pi \sum_{j=1}^{2} \lambda_j \widehat{P}_j{}^{2/\alpha_j} r^{2/\widehat{\alpha}_j} }dr 
\end{equation}
Since the cache-miss users can only associate with the macro tier, the tier
association probability is $\mathcal{P}_{{\rm m},k} = 1$
for $k = 1$ and $ \mathcal{P}_{{\rm m},k} =0$ for $k = 2$.

From the law of total probability, the probability that the
typical user is associated with a MBS or a helper is 
\begin{align}
\mathcal{P}_1 & = p_{\rm h} \mathcal{P}_{{\rm h},1} + (1-p_{\rm h})\mathcal{P}_{{\rm m},1} = p_{\rm h} \mathcal{P}_{{\rm h},1} + 1-p_{\rm h} \label{eqn:ph1}\\
\mathcal{P}_2 & = p_{\rm h} \mathcal{P}_{{\rm h},2} + (1-p_{\rm h})\mathcal{P}_{{\rm m},2} = p_{\rm h} \mathcal{P}_{{\rm h},2} 
\end{align}
From \eqref{eqn:ph1} and using the
conditional probability formula, we can also obtain the probability that a typical
user associated with the MBS is a cache-hit user as $p_{1, {\rm
h}} = {p_{\rm h} \mathcal{P}_{{\rm h},1}}/{\mathcal{P}_1}$,  which is essential for the following
derivation.

Similar to the conventional HetNet, the probability that a BS in the
$k$th tier is active is $ p_{{\rm a}, k} \approx 1 - \big(1 +
\tfrac{\mathcal{P}_k\lambda_u}{3.5\lambda_k}\big)^{-3.5}$.

Considering that the cell association of cache-hit users is based on maximal average receive power (the same as in conventional HetNet), the
average achievable rate of the typical cache-hit user associated with the
$k$th tier can be obtained similarly as we deriving \eqref{eqn:ER1}, which is
\begin{multline}
  \bar R_{{\rm h},k} = \frac{2\pi\lambda_k}{\mathcal{P}_{{\rm h},k}} \int_{0}^{\infty}\!\! \int_{0}^{\infty}\!\! \exp\bigg(- M_k P_k^{-1}r^{\alpha_k} (e^x-1)\sigma^2 \\
  - \pi\sum_{j=1}^{2} \lambda_j\widehat{P}_j{}^{\frac{2}{\alpha_j}}r^{\frac{2}{\widehat{\alpha}_j}}
 ( 1 + p_{a,j} \mathcal{Z}_j(x) ) \bigg)  rdxdr \label{eqn:ERh}
\end{multline}
Different from the conventional pico tier, when a cache-hit user
is associated with a helper, the helper node can fetch the
requested file from its local cache, hence the achievable rate is
no longer limited by the backhaul capacity.

Since cache-miss users can only  associate with the macro tier, the PDF of
the distance between a typical cache-miss user and its serving MBS
 is \cite{tractable}
\begin{equation}
f_{r_1}(r) = 2\pi \lambda_1 r e^{-\lambda_1 \pi r^2} \label{eqn:PDFmiss} 
\end{equation}
Similar to the derivation of \eqref{eqn:ERk} and \eqref{eqn:ERkrk}, we can
obtain the average achievable rate of the typical cache-miss user as 
\begin{equation}
  \bar{R}_m \!\!=\!\!\int\limits_{0}^{\infty} \!\int\limits_{0}^{\infty}\!\!f_{r_1}\!(r) e^{-\frac{M_1 }{P_1} r^{\alpha_1}\!(e^x\!-1)\sigma^2}\!\prod\limits_{j=1}^{2} \!\mathcal{L}_{I_{1,j}} \!\big(\tfrac{ M_1}{P_1} r^{\alpha_1}\!(e^x\!-1)\big) dx dr \nonumber 
\end{equation}
where $\mathcal{L}_{I_{1,j}}$ is given in the following proposition.
\begin{proposition}
	The Laplace transform of the interference from the $j$th tier when the typical cache-miss user is associated with the macro tier, $\mathcal{L}_{I_{1,j}}(s)$,  is
	\begin{equation}
		\left\{ \begin{array}{ll}
		\!\!e^{ -\pi p_{{\rm a}, 1}\lambda_1 r^2 \big(\!{}_2F_1\big[-\frac{2}{\alpha_1},M_1; 1- \frac{2}{\alpha_2}; -\frac{sP_1 }{M_1}r^{-\alpha_1} \!\big]\! - 1\big) },&\!\!\!\!\!j = 1\\
		\!\!e^{ -\pi p_{{\rm a}, 2} \lambda_2   {\Gamma(1-\frac{2}{\alpha_2})\Gamma(M_2 + \frac{2}{\alpha_2}\!)}{\Gamma(M_2)}^{-1}  \big(\frac{sP_2}{M_2}\big)^{\frac{2}{\alpha_2}} },&\!\!\!\!\!j = 2
		\end{array} \right. \label{eqn:Laplacemiss2}
	\end{equation}
\end{proposition}
\begin{proof}
See Appendix B
\end{proof}

Substituting \eqref{eqn:Laplacemiss2} and \eqref{eqn:PDFmiss} into
\eqref{eqn:ERkrk} and setting $C_{\rm bh} \to \infty$, we obtain the average
achievable rate of the typical cache-miss user as,
\begin{align}
& \bar R_{\rm m}     =  2\pi\lambda_1 \int_{0}^{\infty}\!\!\int_{0}^{\infty}\! \exp\Big( -M_1P_1^{-1} r^{\alpha_1}(e^x - 1) \sigma^2   \nonumber\\
&  -\pi \lambda_1 r^2 \left(1 + p_{{\rm a}, 1} \mathcal{Z}_1 (x) \right) -\pi  \lambda_2   r^{\frac{2}{\widehat\alpha_2}}\widehat{P}_2{}^{\frac{2}{\alpha_2}} p_{{\rm a}, 2}\times  \nonumber  \\
&  ~   {\Gamma\big(1-\tfrac{2}{\alpha_2}\big)\Gamma\big(M_2 +  \tfrac{2}{\alpha_2}\big)}{\Gamma(M_2)}^{-1} \widehat{M}_2{}^{-\frac{2}{\alpha_2}} (e^x - 1)^{\frac{2}{\alpha_2}} \Big) r dxdr  \label{eqn:ERm}
\end{align}
where $\mathcal{Z}_1(x)$ is defined in \eqref{eqn:ER2}.

According to the law of total probability, the average cell throughput of a randomly chosen active MBS is 
\begin{equation}
  \bar R_1 = \sum_{n_{\rm h}=0}^{M_1} p_{n_{\rm h}}(n_{\rm h} \bar R_{{\rm h, 1}} + (M_1-n_{\rm h}) \bar R_{m})  \label{eqn:ER1cache}
\end{equation}
where $p_{n_{\rm h}} = \binom{M_1}{n_{\rm h}}p_{1,\rm h}^{k} (1-p_{1, {\rm
h}})^{M_1-k}$ is the probability that $n_k$ users among the $M_1$ random
scheduled users are cache-hit users, and $p_{1, {\rm h}}$ is given after
\eqref{eqn:ph1}, $ n_{\rm h} \bar R_{{\rm h, 1}} + (M_1-n_{\rm h})\bar
R_{m}$ is the average throughput of each macro cell conditioned on that $n_k$
users among the $M_1$  users are cache-hit users.

Since the helper tier can only serve cache-hit users, the average cell throughput of a randomly chosen active helper is $\bar R_2 = \bar R_{\rm h, 2}$. Then, the
ASE of the network can be obtained as
\begin{equation}
  {\rm ASE} = \sum_{j=1}^{2}p_{{\rm a},j}\lambda_j \bar{R}_j \label{eqn:ASE2}
\end{equation}
 In the following, we derive closed-form expressions for approximated $\bar R_{{\rm h},k}$ and $\bar R_{\rm m}$ in a special case.

\subsection{Special Case}
Again, by neglecting thermal noise and considering equal path loss for both tiers, similar to the
derivation of $\eqref{eqn:cER1}$, the average achievable rate of the typical cache-hit user can be approximated
as 
\begin{equation}
\hspace{-0.8mm}   \bar R_{{\rm h},k} \approx \frac{\alpha-2}{2M_k}\mathcal{C}_1 \!\ln \left(\! 1 + \frac{2M_k}{\alpha-2} \!\cdot\! \frac{\ln 2}{\mathcal{C}_1}\right)  +  \frac{\alpha}{2} \mathcal{C}_2 \ln \big(1 + \mathcal{C}_3 4^{-\frac{1}{\alpha}} \big) \hspace{-3.5mm} \label{eqn:cERh} 
\end{equation}
where $\mathcal{C}_1$ is given in \eqref{eqn:cER2},  $\mathcal{C}_2$ and
$\mathcal{C}_3$ are given in \eqref{eqn:cER1}.

Considering $\sigma^2 = 0$,
$\alpha_j= \alpha$ and with a change of variables $r^2 = v$ in
\eqref{eqn:ERm}, $\bar R_{\rm m}$ can be derived as
\begin{multline}
\bar R_{\rm m}    = \lambda_1 \int_{0}^{\infty} \!\Big(  \lambda_1 \left(1 + p_{{\rm a}, 1} \mathcal{Z}_1 (x) \right) +  \lambda_2 \widehat{P}_2{}^{\frac{2}{\alpha}} p_{{\rm a}, 2}\times    \\
    {\Gamma\big(1-\tfrac{2}{\alpha_2}\big)\Gamma\big(M_2 +  \tfrac{2}{\alpha_2}\big)}{\Gamma(M_2)}^{-1} \widehat{M}_2^{-\frac{2}{\alpha_2}}  (e^x - 1)^{\frac{2}{\alpha_2}} \Big)^{-1} dx  \nonumber
\end{multline}

Considering that $\mathcal{Z}_j(x)$ can be approximated as
$\mathcal{Z}_{j,low}(x)$ given in  \eqref{eqn:2F1low} for $x\in[0,\ln2]$ and
as $\mathcal{Z}_{j,high}(x)$ given in  \eqref{eqn:2F1high} for
$x\in[\ln2,\infty)$, the average achievable rate of the typical cache-miss user can be approximated
as  
\begin{align}
\hspace{-1.5mm}\bar R_{\rm m} &   \approx \lambda_1 \int_{0}^{\ln2}\! \big( \lambda_1 ( 1 + p_{{\rm a},1}\mathcal{Z}_{1, low})\big)^{-1} dx \nonumber \\
&  ~~~+\lambda_1 \int_{\ln2}^{\infty} \bigg(\sum_{j=1}^{2} \lambda_j \widehat{P}_j{}^{\frac{2}{\alpha}} \left(1 + p_{{\rm a},j} \mathcal{Z}_{j,high}(x) \right) \bigg) ^{-1}dx \nonumber \\
&    = \!\frac{\alpha -2}{2p_{\rm a, 1}M_1} \ln \left( 1 \!+ \frac{2p_{\rm a, 1}M_1}{\alpha-2}\!\ln 2  \right) +  \frac{\alpha}{2}\mathcal{C}_4 \! \ln\!  \big(1 \! + \mathcal{C}_3 4^{-\frac{1}{\alpha}} \big)  \hspace{-1.5mm}\label{eqn:cERm}
\end{align}
where we also use the approximation $e^x - 1 = \mathcal{O}(x) \approx 0$ for
$x\in[0,\ln2]$ and  $e^x - 1 \approx e^x$ for $x\in[\ln2, \infty)$, and $\mathcal{C}_4 \triangleq 1 / \big( \sum_{j=1}^{2} (1-p_{{\rm
		a},j}) \lambda_j \widehat{P}_j{}^{\frac{2}{\alpha}} \big)$.

By substituting \eqref{eqn:cERh} and \eqref{eqn:cERm} into \eqref{eqn:ASE2}, we can obtain the closed-form expression of an approximated ASE, which is accurate as shown by
simulation later.

\section{Numerical and Simulation Results}
In this section, we validate the analytical results via simulations and compare the performance of conventional
and cache-enabled HetNets by numerical results.

The simulation parameters are given in Table
\ref{tab:para}. To reflect the impact of the file catalog size, we use normalized cache capacity $\eta =
\frac{N_c}{N_f}$ in the following.

\begin{table}[!htb]
	\centering
	\caption{Simulation parameters}
	\begin{tabular}{ll}
		\toprule
		 Parameters & Value \\
		\midrule
		 MBS density, $\lambda_1$ & $1/(500^2 \pi)$ m$^{-2}$ \cite{flexible}\\
		 User density, $\lambda_u$ & $50/(500^2 \pi)$ m$^{-2}$ \\
		 Path loss exponent, $\alpha$ & 3.7 \\
		 Transmit power of each MBS, $P_1$    & 46 dBm\\
		 Transmit power of each PBSs/helper node, $P_2$     & 21 dBm\\
		 Number of antennas at each MBS, $M_1$ & 4\\
		 Transmission bandwidth, $W$  & 20 MHz \\
		 File catalog size, $N_f$ & $10^5$ files \cite{Andy2012} \\
		 Zipf-like distribution skew parameter, $\delta$ & 0.8 \cite{woo2013comparison} \\
		\bottomrule
	\end{tabular}%

	\label{tab:para}%
\end{table}%

\begin{figure}[!htb]
	\centering
	\subfigure[ASE v.s. helper/PBS density $\lambda_2$.]{
	\label{fig:ASE_lambda}
	\hspace{0mm}\includegraphics[width=0.40\textwidth]{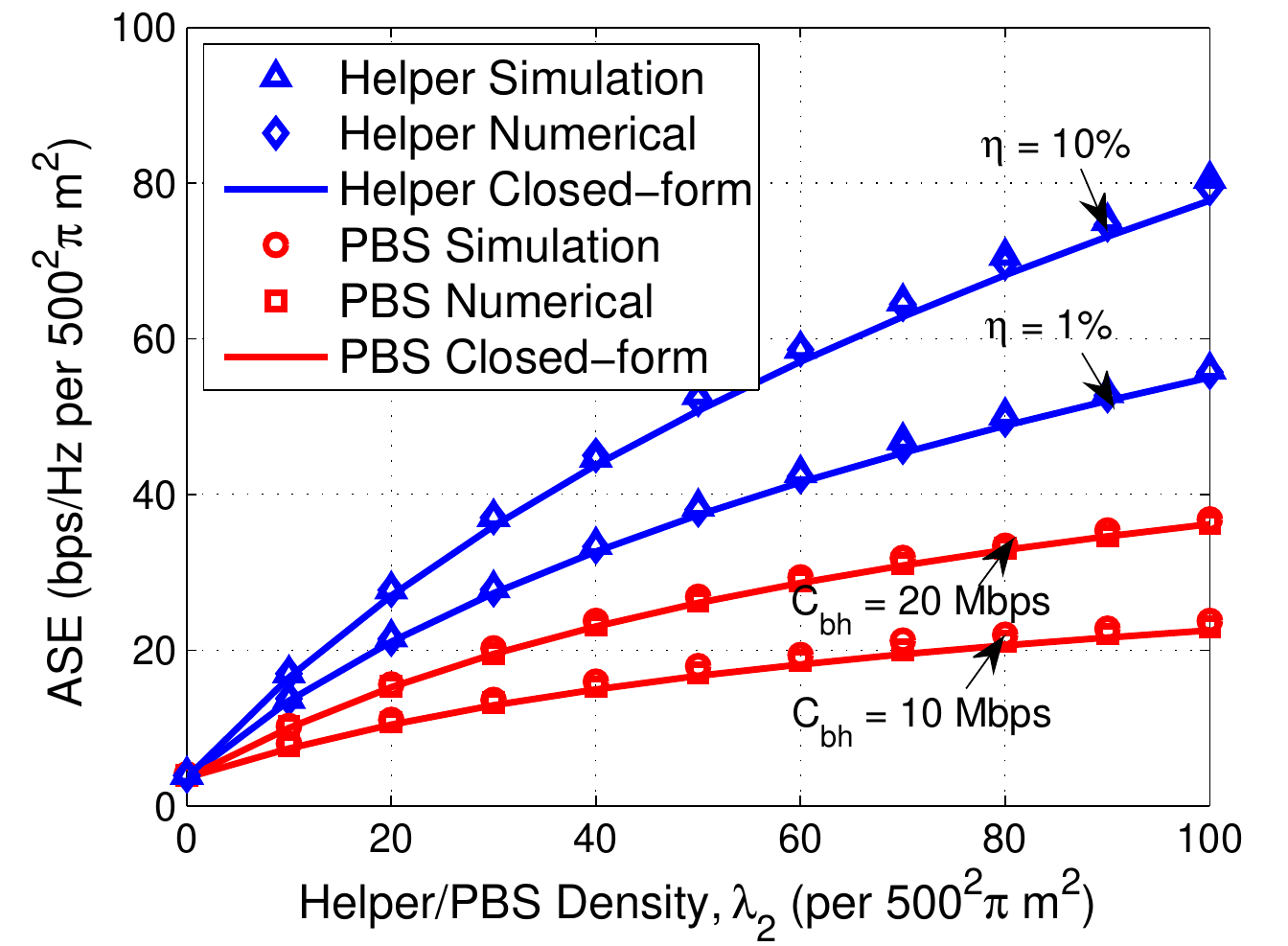}\hspace{0mm}}
	\subfigure[ASE v.s. $\eta$, $\lambda_2 = 50/(500^2\pi)$.]{
	\label{fig:ASE_eta}
	\hspace{0mm}\includegraphics[width=0.40\textwidth]{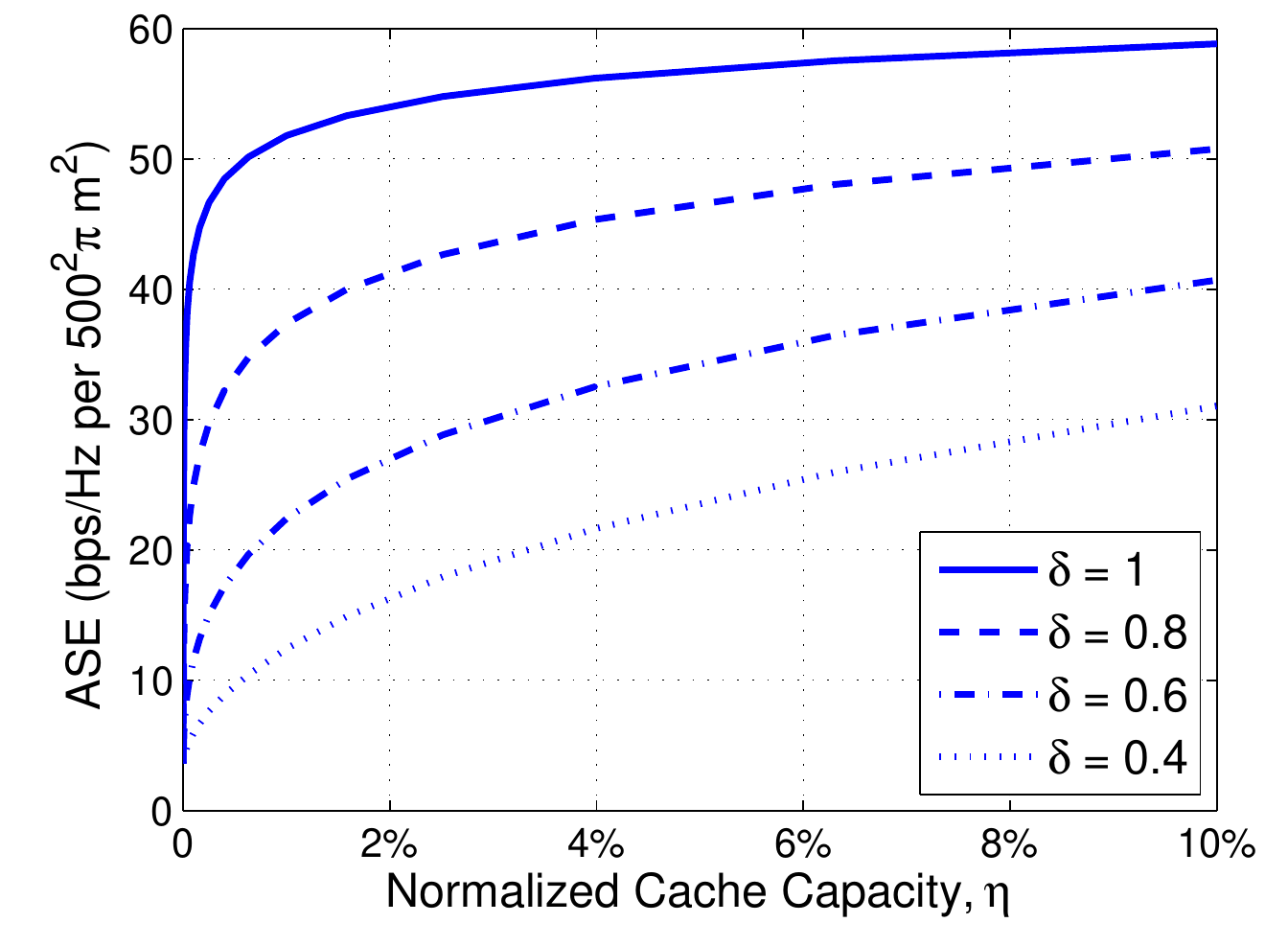} \hspace{0mm}}
	\caption{ASE v.s. helper/PBS density and normalized cache capacity.}
	\label{fig:ASE} 
\end{figure}

In Fig. \ref{fig:ASE_lambda}, we compare the ASE of the two HetNets, where the numerical results are obtained from substituting \eqref{eqn:ER2} and \eqref{eqn:ER1} into \eqref{eqn:ASE1} for conventional HetNet and substituting \eqref{eqn:ERh} and \eqref{eqn:ERm} into \eqref{eqn:ASE2} with \eqref{eqn:ER1cache} for cache-enabled HetNet, the results of the closed-form expression are obtained from substituting \eqref{eqn:cER2} and \eqref{eqn:cER1} into \eqref{eqn:ASE1} for conventional HetNet and substituting \eqref{eqn:cERh} and \eqref{eqn:cERm} into \eqref{eqn:ASE2} with \eqref{eqn:ER1cache} for cache-enabled HetNet. We can see
that the results obtained from closed-form expressions are very close to the numerical and
simulation results. In fact, same conclusion can be obtained from the simulation results using typical values of $\alpha_j$ for different tiers, which are not shown for conciseness. Note that in the simulation, $\sigma^2 \neq 0$. These indicate that the approximations are very accurate even without the assumption of $\sigma^2 = 0$,
$\alpha_j= \alpha$. Hence, in the sequel we only provide the analytical results obtained from the closed-form expressions.
Compared with the conventional HetNet with limited-capacity backhaul (e.g., $C_{\rm bh}=10$
Mbps), the cache-enabled HetNet can double the ASE when each helper node only
caches 1\% of the total files. Alternatively, to achieve the same ASE, the
helper node density is much lower than the PBS density, which can
reduce the deploying and operating cost remarkably (e.g., when
$C_{\rm bh} = 10$ Mbps and $\eta = 1\% $, the helper node density is
about 1/4 of the PBS density to achieve an ASE of $20/(500^2\pi)$ bps/Hz/m$^2$). As expected, the ASE of cache-enabled HetNet increases with $\eta$. Furthermore, when $\delta$ is larger, the ASE grows with $\eta$  more rapidly for small value of $\eta$ and grows with $\eta$ more slowly for large $\eta$.

\begin{figure}[!htb]
	\centering
	\includegraphics[width=0.42\textwidth]{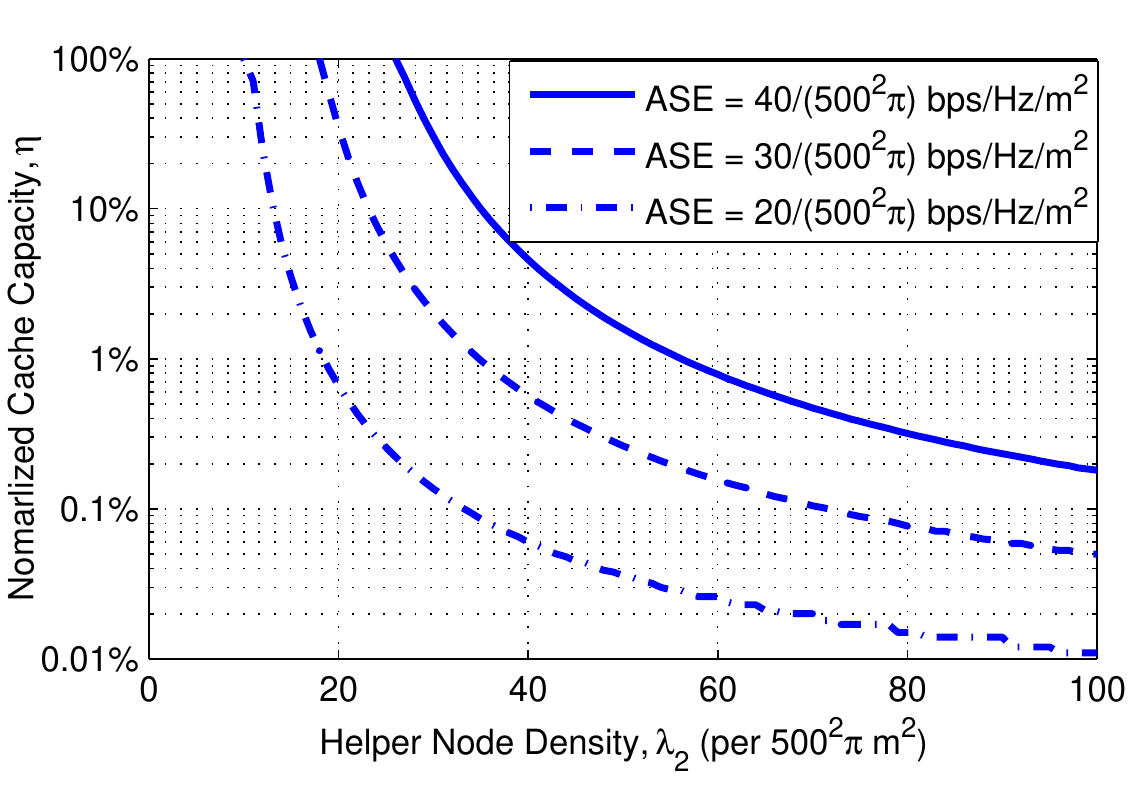}
	\caption{Trade-off between helper density and cache capacity. }
	\label{fig:tradeoff} 
\end{figure}

Since the ASE of cache-enabled HetNet can be improved either by increasing
cache capacity or increasing helper density, a natural question is that how much helper
density can be traded off by cache capacity to achieve a target
ASE? To answer this question, we set the ASE as different values and show
the normalized cache capacity versus helper density in Fig.
\ref{fig:tradeoff}. With a given target ASE and helper density, $\eta$ can be found by substituting \eqref{eqn:ph},
\eqref{eqn:ER1cache} into \eqref{eqn:ASE2} and then using the bisection
search method. It is shown that we can reduce the helper
density by increasing the cache capacity of each helper. For example,
to achieve a target ASE of $20/(500^2 \pi)$ bps/Hz/m$^2$, by increasing
the cache capacity from $\eta = 0.01\%$ to $\eta =
0.1\%$, the helper density can be reduced by two thirds. Similar trade-off between BS density and cache
capacity was reported in \cite{EURASIP}, where
a homogeneous network with single antenna BSs was considered and the performance metric was outage probability. In our work, the helper density can be traded off more significantly by the cache capacity.

\begin{figure}[!htb]
	\centering
	\includegraphics[width=0.42\textwidth]{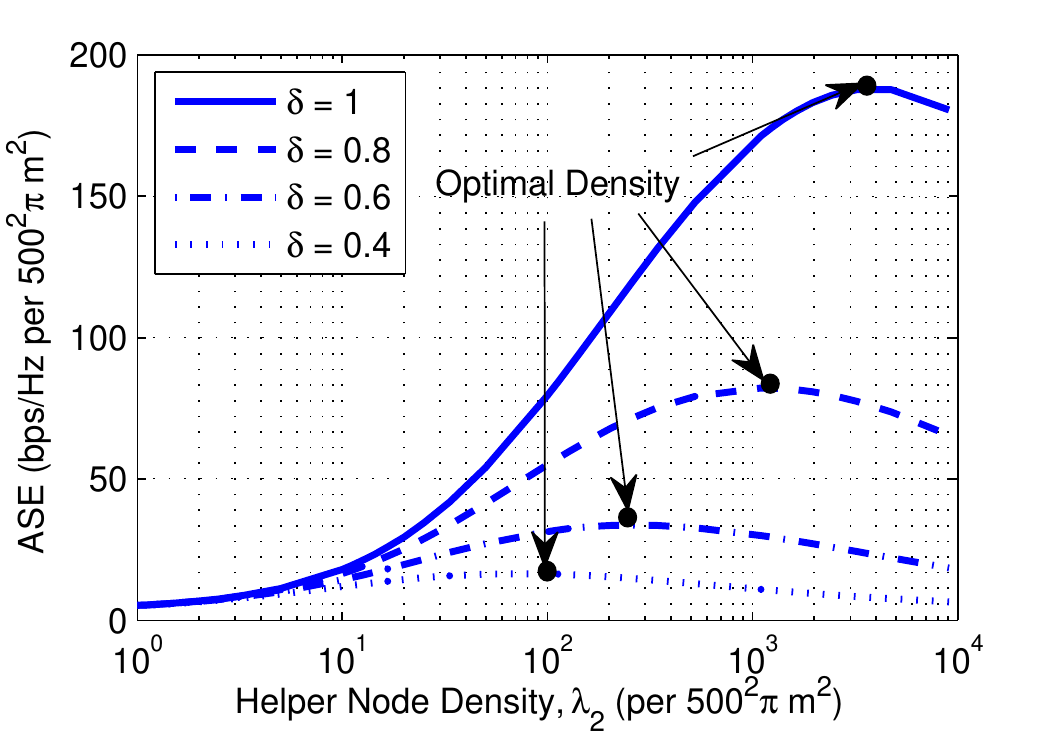}
	\caption{ASE v.s. helper  density with given $\lambda_2N_c  = 10^4/(500^2\pi)$.}
	\label{fig:delta} 
\end{figure}

Inspired by such a trade-off, another natural question is: with a given total amount of cache capacity
within an area, should we deploy the caches in a distributed manner (i.e., more helpers each with less cache capacity) or in a centralized manner (i.e.,
less helpers each with more cache capacity) in order to maximize
the ASE? To answer this question, we fix the area cache capacity
$\lambda_2N_c $ as a constant and provide the ASE versus the helper
density in Fig. \ref{fig:delta}. We can see that there exists an
optimal helper density maximizing the ASE. Moreover, the optimal  density increases with $\delta$, which means that the more skewed the file
popularity is, the more distributedly we should deploy the caches.
This can be explained from the impact of the following two observations in Figs. \ref{fig:ASE_lambda} and \ref{fig:ASE_eta}. On one
hand, with given cache capacity of each helper, the ASE increases with
the helper density first rapidly and then slowly. On the other hand, with
given helper density, the ASE reduces with the decrease of cache
capacity first slowly and then rapidly, and the larger $\delta$ is, the more slowly the
ASE decreases with $\eta$ when the cache capacity is large.
\section{Conclusion}
In this paper, we investigated the gain of cache-enabled HetNet over
conventional HetNet with limited capacity backhaul, and addressed the tradeoff
in deploying the cache-enabled HetNet. We obtained closed-form expressions of the approximated ASEs of these two HetNets.
We then used numerical results to show that cache-enabled HetNet can double
the ASE over conventional HetNet with the same BS/helper density. To achieve the same ASE, the helper density is much lower than the PBS density and can be
traded off by cache capacity. For a given total cache capacity within an
area, there exists an optimal helper  density maximizing the ASE.

\appendices 
\section{Proof of Proposition 1} 
The Laplace transform of $I_{k,j}$ can be derived as 
\begin{align}
&   \mathcal{L}_{I_{k,j}} (s) = \mathbb{E}_{\tilde \Phi_j, h_{j,i}} \left[e^{-s \sum_{i\in \tilde\Phi_j} P_j  h_{j,i} r_{j,i}^{-\alpha_j}}\right] \nonumber \\
&  \overset{(a)}{=}  \mathbb{E}_{\tilde \Phi_j} \left[ \prod_{i\in \tilde\Phi_j}\left(1 + \frac{s P_j}{M_j}  r_{j,i}^{-\alpha_j}  \right)^{-M_j}\right]  \nonumber \\
&   \overset{(b)}{=} \exp\left( -2\pi p_{{\rm a}, j}\lambda_j \int_{r_{0,j}}^{\infty} \left(1 - \left({1 + \frac{s P_j}{M_j} u^{-\alpha_j}}  \right)^{-M_j}\right)u du\right) \nonumber\\
&   \overset{(c)}{=} \exp\left(\!-\pi p_{{\rm a}, j}\lambda_j  \int_{r_{0,j}^2}^{\infty}\left(1 - \left( 1  + \frac{s P_j}{M_j} v^{-\frac{\alpha_j}{2}}  \right)^{-M_j}\right) dv \right) \label{eqn:Laplace 1}
\end{align}
 where step $(a)$ follows from $h_{j,i}\sim \mathbb{G}(M_j, \frac{1}{M_j})$,
step $(b)$ follows from the probability generating function of the PPP, and
step $(c)$ is obtained by changing variables as $u^2 =v $.

To derive the integration in \eqref{eqn:Laplace 1}, we first obtain the indefinite integration as 
\begin{align}
&   \int \left(1 - \left( 1  + \frac{s P_j}{M_j} v^{-\frac{\alpha_j}{2}}  \right)^{-M_j}\right)dv  \nonumber \\
&   \overset{(a)}{=} \int dv - \int\sum_{n = 0}^{\infty}  \frac{(-1)^n (M_j)_n}{n!} \left(\frac{sP_j}{M_j} v^{-\frac{\alpha_j}{2}}\right)^{n} dv  \nonumber \\
&   = v - \sum_{n = 0}^{\infty}\frac{(-1)^n (M_j)_n}{n!} \int\left(\frac{sP_j}{M_j} v^{-\frac{\alpha_j}{2}}\right)^{n} dv \nonumber\\
&   = v - v\sum_{n = 0}^{\infty} \frac{ (-\frac{2}{\alpha_j})_n (M_j)_n}{(1-\frac{2}{\alpha_j})_n} \frac{\left(-\frac{sP_j}{M_j} v^{-\frac{\alpha_j}{2}}\right)^{n}}{n! } \nonumber\\
&   \overset{(b)}{=} v\big(1 - {}_2F_1\big[-\tfrac{2}{\alpha_j}, M_j; 1-\tfrac{2}{\alpha_j}; -\tfrac{sP_j }{M_j}v^{-\frac{\alpha_j}{2}}\big]\big)
\end{align}
where step (a) is from the generalized binomial theorem,  $(x)_n \triangleq
x(x+1) \cdots (x + n - 1)$  denotes the rising Pochhammer symbol, step (b)
is from the series-form expression of Gauss hypergeometric function
${}_2F_1[\cdot]$.

By introducing the integration limits and after some
manipulations, we obtain 
\begin{multline}
   \int_{r_{0,j}^2}^{\infty} \left(1 - \left( 1  + \frac{s P_j}{M_j} v^{-\frac{\alpha_j}{2}}  \right)^{-M_j}\right) dv \\
  = r_{0,j}^2 \big({}_2F_1 \big[-\tfrac{2}{\alpha_j}, M_j; 1-\tfrac{2}{\alpha_j}; -\tfrac{sP_j }{M_j}r_{0,j}^{-\alpha_j}\big] - 1\big) \label{eqn:integral}
\end{multline}

The lower limit of the
integration is $r_{0,j} = \widehat{P}_j{}^{\frac{1}{\alpha_j}}r^{\frac{1}{\widehat{\alpha}_j}}$,
which is the possibly  closest distance of the interfering BS in the $j$th
tier.
By considering the integration  limit $r_{0,j}$ and substituting  \eqref{eqn:integral} into
\eqref{eqn:Laplace 1}, the proposition can be proved.

\section{Proof of Proposition B}
Similar to the derivation of \eqref{eqn:Laplace}, we obtain 
\begin{equation}
\mathcal{L}_{I_{1,j}}(s) \! =e^{\! -\pi p_{{\rm a}, j}\lambda_j  r_{0,j}^2 \big(\!{}_2F_1 \big[-\frac{2}{\alpha_j}, M_j; 1-\frac{2}{\alpha_j}; -\frac{sP_j }{M_j}r^{-\alpha_j}\!\big] \!- \!1\big) } \label{eqn:Laplacemiss} 
\end{equation}	
where  $r_{0,j} = r$ for $j =1 $ and $r_{0,j} \to 0$ for $j=2$ (helper tier). The proposition can be proved by substituting $r_{0,j}$ and considering
$\lim_{r_{0,j}\to 0} r_{0,j}^2 \big({}_2F_1
\big[-\tfrac{2}{\alpha_j}, M_j; 1-\tfrac{2}{\alpha_j}; -\tfrac{sP_j
	}{M_j}r_{0,j}^{-\alpha_j}\big] - 1\big) =
	{\Gamma\big(1-\frac{2}{\alpha_j}\big)\Gamma\big(M_1 +
		\frac{2}{\alpha_j}\big)}  {\Gamma(M_j)}^{-1}
		\big(\frac{sP_j}{M_j}\big)^{\frac{2}{\alpha_j}}$ derived from
		\eqref{eqn:trans}.
		
\bibliographystyle{IEEEtran}
\bibliography{dongbib}
\end{document}